\documentclass[12pt]{article}
\usepackage{localeng}
\usepackage{geometry}
\DeclareMathOperator{\KS}{\mathrm{C}\mskip 1mu}
\DeclareMathOperator{\edim}{\textit{dim}\mskip 1 mu}
\newcommand{\cnd}{\mskip 1mu | \mskip 1mu }
\newtheorem{theorem}{Theorem}
\newtheorem{proposition}{Proposition}
\newtheorem*{lemma}{Lemma}

\begin{document}
\title{Inequalities for entropies and dimensions}
\author{Alexander Shen\thanks{LIRMM, University of Montpellier, CNRS, Montpellier, France, \texttt{alexander.shen@lirmm.fr}, \texttt{sasha.shen@gmail.com}.}}
\date{}
\maketitle
\begin{abstract}
We show that linear inequalities for entropies have a natural geometric interpretation in terms of Hausdorff and packing dimensions, using the point-to-set principle and known results about inequalities for complexities, entropies and the sizes of subgroups.
\end{abstract}

\section{Introduction}

\subsection*{Inequalities for entropies}

Let $\xi_1,\ldots,\xi_m$ be jointly distributed random variables with finite ranges. Then, for every nonempty $I\subset\{1,\ldots,m\}$, we may consider the tuple of variables 
\[
\xi_I=\langle \xi_i \mid i\in I\rangle,
\] 
and its \emph{Shannon entropy} $H(\xi_I)$. Recall that the Shannon entropy of a random variable $\xi$ that takes $s$ values with probabilities $p_1,\ldots,p_s$ is defined as $H(\xi)=\sum_{i} p_i\log (1/p_i)$. In this way we get $2^m-1$ real numbers (for $2^m-1$ non-empty subsets of $\{1,\ldots,m\}$). Shannon pointed out some inequalities that are always true for those quantities. For example, for every two variables $\xi_1,\xi_2$ we have 
\[
H(\xi_1)\le H(\xi_1,\xi_2)\le H(\xi_1)+H(\xi_1),
\]
and for every triple of variables $\xi_1,\xi_2,\xi_3$ we have 
\[
H(\xi_1)+H(\xi_1,\xi_2,\xi_3)\le H(\xi_1,\xi_2)+H(\xi_1,\xi_3).
\]
The latter inequality corresponds to the inequality
\[
H(\xi_2,\xi_3 \cnd \xi_1)\le H(\xi_2\cnd \xi_1)+H(\xi_3\cnd\xi_1)
\]
for conditional entropies (defined as $H(\xi\cnd\eta)=H(\xi,\eta)-H(\eta)$).

For a long time no other valid inequalities for entropies (except positive linear combinations of Shannon's inequalities) were known. Then Zhang and Yeung~\cite{zhang-yeung} found an inequality that is not a positive linear combination of Shannon's inequalities, and since then a lot of other inequalities of this type (\emph{non-Shannon inequalities}) were found. It became clear that the set of all valid linear inequalities for entropies has a complex structure (see, e.g.,~\cite{matus-csirmaz}). On the other hand, it became also clear that this set is very fundamental since it can be equivalently defined in combinatorial terms, in terms of subgroups size, and in terms of Kolmogorov complexity (see~\cite[Chapter 10]{suv} for the historical account and references). In this paper we use these characterizations to get one more equivalent characterization of this set, now in terms of Hausdorff and packing dimensions. 

Let us recall briefly the characterizations in terms of Kolmogorov complexity and subgroup sizes; more details and proofs can be found in~\cite[Chapter~10]{suv}.

\subsection*{Inequalities for complexities and subgroup sizes}

For a binary string $x$, its \emph{Kolmogorov complexity} is defined as the minimal length of a program that produces this string. The definition depends on the choice of programming language, but (as Kolmogorov and Solomonoff noted) there exist \emph{optimal} programming languages that make the complexity function minimal up to an additive constant, and we fix one of them. In this way we define the function $\KS(x)$ up to a bounded additive term, so to get a meaningful statement we should consider the asymptotic behavior of complexity when the length of the strings goes to infinity. (For more details and proofs about Kolmogorov complexity see, e.g.,~\cite{suv}.) 

Let $x_1,\ldots,x_m$ be binary strings. A tuple $\langle x_1,\ldots,x_m\rangle$ of strings can be computably encoded as one string whose complexity is (by definition) the complexity of the tuple $\langle x_1,\ldots,x_m\rangle$ and is denoted by $\KS(x_1,\ldots,x_m)$. The change of the encoding changes the complexity of a tuple at most by $O(1)$, so this notion is well defined. As it was noted by Kolmogorov~\cite{kolm1968}, some inequalities for Shannon entropies have Kolmogorov complexity counterparts. For example, 
\[
\KS(x_1)+\KS(x_1,x_2,x_3)\le \KS(x_1,x_2)+\KS(x_1,x_3)+O(\log n)
\]
for all strings $x_1,x_2,x_3$ of length at most $n$; this statement corresponds to the inequality for entropies mentioned earlier. Note that this inequality is asymptotic; since $\KS(x)$ is defined up to $O(1)$ additive terms, some error term is unavoidable. We use $O(\log n)$ term instead of $O(1)$, so the difference between different versions of Kolmogorov complexity (e.g., plain and prefix complexity) does not matter for us.

It is easy to see that every linear inequality for complexities that is true for complexities with logarithmic precision, is also true for Shannon entropies: we replace strings by random variables, Kolmogorov complexity by entropy and omit the error term. It was proven by Romashchenko~\cite{romash} that the reverse implication is also true and the same linear inequalities are valid for entropies and complexities (the exact statement and the proof can be found in~\cite[Chapter 10]{suv}). This result remains valid if we allow the programs in the definition of complexity access some oracle (set of strings) $X$ (this is called \emph{relativization} in computability theory); in particular, the class of valid linear inequalities for complexities with oracle $X$ does not depend on $X$ (since it coincides with the class of valid inequalities for entropies).

Another characterization of the same class of inequalities (in terms of sizes of subgroups of some group and their intersections) was given by Zhang and Yeung~\cite{zhang-yeung}. They have shown that it is enough to consider some special type of random variables that correspond to a group and its subgroups: if a linear inequality is valid for random variables of this type, it is valid for all random variables. (See below Section~\ref{equivalence} for more details about this class.)

\subsection*{Dimensions and point-to-set principle}

The notions of Hausdorff dimension and packing dimension for sets in $\mathbb{R}^m$ are well known in the geometric measure theory, but for our purposes it is convenient to use their equivalent definitions provided by the point-to-set principle formulated by Jack Lutz and Neil Lutz~\cite{lutz}; the references to the classical definitions (and to previous results about effective dimensions) can be found there.

Let $\alpha=0.a_1a_2a_3\ldots$ be a real number represented as a binary fraction (the integer part does not matter, and we assume that $\alpha\in[0,1]$). Consider the Kolmogorov complexity of the first $n$ bits of $\alpha$ and then consider the limits
\[
\edim_H(\alpha) = \liminf_{n\to\infty} \frac{\KS(a_1\ldots a_n)}{n} 
\qquad\text{and}\qquad
\edim_p(\alpha) = \limsup_{n\to\infty} \frac{\KS(a_1\ldots a_n)}{n}
\]
called \emph{effective Hausdorff dimension} and \emph{effective packing dimension} of $\alpha$. Note that the classical dimension of every point is zero, so these notions do not have classical counterparts. Both dimensions (for every real $\alpha$) are between $0$ and $1$.

Now we switch from points to sets and for every set $A\subset[0,1]$ we define the \emph{effective Hausdorff} and \emph{effective packing} dimension  of $A$ as the supremum of corresponding dimensions of the points in $A$:
\[
\edim_H(A)=\sup_{\alpha\in A} \edim_H(\alpha)
\qquad \text{and} \qquad 
\edim_p(A)=\sup_{\alpha\in A} \edim_p(\alpha).
\]
This definition can be relativized by some oracle $X$; for that we replace the Kolmogorov complexity $\KS$ by its relativized version $\KS^X$. Adding oracle can make Kolmogorov complexity and effective dimension smaller; we denote the relativized effective dimensions by $\edim^X_H(A)$ and $\edim^X_p(A)$. The \emph{point-to-set principle} says that for every set $A$ there exists an oracle $X$ that makes the effective dimensions minimal and these minimal dimensions are classical Hausdorff and packing dimensions:
\[
\dim_H(A)=\min_{X} \edim^X_H(A)
\qquad\text{and}\qquad
\dim_p(A)=\min_{X} \edim^X_p(A)
\]
Note that we use ``$\edim$'' (italic) for effective dimensions and ``$\dim$'' for classical dimensions to distinguish between them (for sets; for individual points only effective dimensions make sense). We take this characterization as (equivalent) definition of Hausdorff and packing dimensions. 

Formally speaking, we should first prove that the minimal values are achieved for some oracle $X$; this is easy to see, because a countable sequence of oracles that give better and better approximations can be combined into one oracle. (Or we could just use ``$\inf$'' instead of ``$\min$'' in the definition.)

In the same way the dimensions of a set $A\subset \mathbb{R}^m$ are defined. Now a point $\alpha \in A$ has $m$ coordinates and is represented by a tuple $\langle\alpha_1,\ldots,\alpha_m\rangle$ of binary fractions. To define the effective dimension of $\alpha$ we now consider the Kolmogorov complexity of an $m$-tuple that consists of $n$-bit prefixes of $\alpha_1,\ldots,\alpha_m$, and divide this complexity by $n$. Taking the limits, we get the effective Hausdorff and packing dimension of the point $\alpha$; they are between $0$ and $m$. Now, taking supremum over all $\alpha$ in $A$, we define the effective dimensions of $A\subset\mathbb{R}^m$, and their relativized versions $\edim^X_H(A)$ and $\edim^X_p(A)$ are defined in a similar way. Taking the minimum over $X$, we get the classical Hausdorff and packing dimensions of the set $A$.

In the next section we give two examples where the inequalities for Kolmogorov complexities are used to prove some results about (classical) Hausdorff and packing dimensions. Then (Section~\ref{sec:general}) we generalize this approach to arbitrary inequalities for Kolmogorov complexity. Finally (Section~\ref{equivalence}) we prove the reverse statement that shows that this translation can be used to characterize exactly the class of all linear inequalities valid for entropies or complexities.

\section{Inequalities for dimensions: examples}

\subsection*{Example 1}

Consider the inequality for Kolmogorov complexities
\begin{equation}
2\KS(x,y,z) \le \KS(x,y)+\KS(x,z)+\KS(y,z)+O(\log n) \label{ineq1a}
\end{equation}
that is true for all strings $x,y,z$ of length $n$ (see, e.g.,~\cite[Section 2.3]{suv}). We apply it to the first $n$ bits of three real numbers  $\alpha,\beta,\gamma$ (considered as binary sequences; as we have said, we ignore the integer part and assume that all the real numbers are between $0$ and $1$).
\[
2 \KS((\alpha)_n, (\beta)_n, (\gamma)_n)\le 
\KS((\alpha)_n, (\beta)_n)+
\KS((\alpha)_n,  (\gamma)_n)+
\KS((\beta)_n, (\gamma)_n) + O(\log n).
\]
Here we denote by $(\rho)_n$ the first $n$ bits of a real number $\rho$ considered as a binary sequence. We can divide this inequality by $n$ and take $\limsup$ of both parts: recall that $\limsup (x_n+y_n)\le \limsup x_n+\limsup y_n$. In this way we get the inequality
\begin{equation}
2\edim_p (\langle\alpha,\beta,\gamma\rangle)\le \edim_p(\langle\alpha, \beta\rangle)+\edim_p(\langle\alpha,\gamma\rangle)+\edim_p(\langle\beta,\gamma\rangle).\label{ineq1b}
\end{equation}
This inequality can be relativized with arbitrary $X$ used as an oracle in the definitions of Kolmogorov complexity and effective packing dimension.

Now instead of one point $\langle \alpha,\beta,\gamma\rangle$ consider a set $S\subset \mathbb{R}^3$. Consider also three its two-dimensional projections onto each of three coordinate planes; we denote them by $S_{12}$, $S_{13}$ and $S_{23}$. The point-to-set principle says that the (classical) packing dimension of a set is the minimal (over all oracles) \emph{effective} packing dimension of the set relativized to the oracle, and the latter is the supremum (over all points in the set) of effective packing dimensions of its points. Fix an oracle $X$ that makes the effective packing dimension of all three projections minimal (we may combine the oracles for all three sets) and use it everywhere. Then for every point $s=\langle s_1,s_2,s_3\rangle$ in $S$ the effective packing dimension of $\langle s_1,s_2\rangle$ does not exceed the packing dimension of $S_{12}$, etc. Applying the inequality~(\ref{ineq1b}), we conclude that effective packing dimension of every point $\langle s_1,s_2,s_3\rangle\in S$ satisfies the inequality
\begin{equation*}
2\edim_p^X(\langle s_1,s_2,s_3\rangle)\le \dim_p S_{12}+\dim_p S_{13} + \dim_p S_{23}. \label{ineq1c}
\end{equation*}
where the effective dimension in the left hand side is taken with the fixed oracle (and classical dimensions in the right hand side do not depend on any oracles). Now we take the maximum of all points $\langle s_1,s_2,s_3\rangle\in S$ and get the bound for effective packing dimension of $S$ (with oracle $X$), and, therefore, for the classical packing dimension of $S$. In this way we get the following result that deals exclusively with classical dimensions:

\begin{proposition}\label{ineq1d}
For every set $S\subset \mathbb{R}^3$ and three its two-dimensional projections $S_{12}$, $S_{13}$ and $S_{23}$ we have
\begin{equation*}
2\dim_p S \le \dim_p S_{12} + \dim_p S_{13} + \dim_p S_{23}. 
\end{equation*}
\end{proposition}

\subsection*{Example 2}

Now consider another inequality for Kolmogorov complexities of three strings mentioned earlier:
\begin{equation}
\KS(x)+\KS(x,y,z)\le \KS(x,y) + \KS(x,z)+O(\log n). \label{ineq2a}
\end{equation}
We can try the same reasoning, but some changes are necessary. First, we use that 
\[
\liminf x_n+\liminf y_n \le \liminf (x_n+y_n)
\]
(and that $\liminf\le \limsup$) to get an inequality that combines the effective Hausdorff and effective packing dimensions:
\begin{equation}
\edim^X_H(s_1)+\edim^X_H(\langle s_1,s_2,s_3\rangle)\le \edim^X_p (\langle s_1,s_2\rangle)+\edim^X_p (\langle s_1,s_3\rangle), \label{ineq2b}
\end{equation}
for every oracle $X$. Then, for some strong enough oracle~$X$, we get
\begin{equation*}
\edim^X_H(s_1)+\edim^X_H(\langle s_1,s_2,s_3\rangle)\le \dim_p S_{12}+\dim_p S_{13}. 
\end{equation*}
Still we cannot make any conclusions about the classical dimensions of the projection $S_1$ and the entire set $S$, since the point $s_1$ where the first term is maximal is unrelated to the point $s=\langle s_1,s_2,s_3\rangle$ when the second term is maximal. 

To see what we can do, let us recall that a similar problem appears for the combinatorial interpretation of inequalities for Kolmogorov complexity (see~\cite[Chapter 10]{suv} for details). The inequality~(\ref{ineq1a}) from our previous example has a direct combinatorial translation (a special case of the Loomis--Whitney inequality): if a three-dimensional body has volume\footnote{To be more combinatorial, one could consider finite sets and the cardinalities of those sets and their projections.} $V$ and its three projections have areas $V_{12}$, $V_{13}$ and $V_{23}$, then 
\begin{equation*}
2 \log V \le \log V_{12}+ \log V_{13}+\log V_{23} \qquad (\text{or } V^2 \le V_{12}V_{13}V_{23}). \label{ineq1e}
\end{equation*}
However, the inequality 
\[
\log V_1 + \log V \le \log V_{12}+ \log V_{13},
\]
written in a similar way for the inequality~(\ref{ineq2a}), does not always hold for a three-dimensional body. Consider, for example, the union of a cube $N\times N\times N$ with a parallelepiped $N^{1.5}\times 1\times 1$: the left hand side is about $1.5\log N+3\log N = 4.5\log N$, while the right hand side is about $2\log N + 2 \log N = 4\log N$ (for large $N$).

The solution for the combinatorial case is to allow splitting of the set $V$ into two parts: one has (relatively) small projection length, the other has (relatively) small volume. Namely, the following statement is true (see~\cite[Section 10.7]{suv}):
\begin{quote}
if for some three-dimensional set $S$ the areas $V_{12}$ and $V_{13}$ of its two-dimensional projections onto coordinates $(1,2)$ and $(1,3)$ satisfy the inequality 
\[
\log V_{12} + \log V_{13} \le a+b, 
\]
then the set can be split in two parts
\[
S=S'\cup S''
\]
in such a way that
\[
\log V'_1 \le a\quad \text{and}\quad  \log V'' \le b.
\]
Here $V'_1$ is the measure of the (one-dimensional) projection of $S'$ onto the first coordinate, and $V''$ is the volume of $S''$.
\end{quote}
(Why do we introduce $a$ and $b$? In a sense, we replace the inequality $u+v\le w$ by an equivalent statement ``for every $a$, $b$, if $w\le a+b$, then either $u\le a$ or $w\le b$''.)
\smallskip

We use a similar approach for dimensions, and get the following statement.
\begin{proposition}\label{split1}
Let $S\subset \mathbb{R}^3$, and let $a,b\ge 0$ be two numbers such that
\[
\dim_p S_{12} + \dim_p S_{13} \le a+b.
\]
Then there exist a splitting $S=S'\cup S''$ such that
\[
\dim_H (S'_1)\le a \quad \text{and} \quad \dim_H (S'')\le b.
\]
\end{proposition}
Here $S'_1$ is the projection of $S'$ onto the first coordinate.

\begin{proof}
As before, fix  an oracle $X$ that minimizes the effective packing dimensions of $S_{12}$ and $S_{13}$, and use it everywhere when speaking about complexities and effective dimensions. Then for every point $\langle s_1,s_2,s_3\rangle\in S$ we have 
\[
\edim^X_p (\langle s_1,s_2\rangle)+\edim^X_p(\langle s_1,s_3\rangle) \le a+b.
\] 
The inequality~(\ref{ineq2b}) then guarantees that
\[
\edim^X_H(s_1)+\edim^X_H(\langle s_1,s_2,s_3\rangle)\le a+b
\]
for every $\langle s_1,s_2,s_3\rangle \in S$, and therefore 
\begin{center}
either $\edim^X_H(s_1)\le a$ or $\edim^X_H(\langle s_1,s_2,s_3\rangle )\le b$ 
\end{center}
for every $\langle s_1,s_2,s_3\rangle \in S$. Therefore we may split $S$ into two sets $S'$ and $S''$ and guarantee that for all elements $\langle s_1,s_2,s_3\rangle\in S'$ we have $\edim^X_H(s_1)\le a$ and for all elements $\langle s_1,s_2,s_3\rangle\in S''$  we have $\edim^X_H(\langle s_1,s_2,s_3\rangle)\le b$.  This implies that $\dim_H(S'_1)\le a$ and $\dim_H(S'')\le b$, as required.
\end{proof}

\section{Corollaries for dimensions: general statement}\label{sec:general}

Proposition~\ref{split1} can be generalized (with essentially the same proof) to arbitrary linear inequalities for Kolmogorov complexities. Fix some $m$, and consider an $m$-tuple of strings $\langle x_1,\ldots,x_m\rangle$. For every non-empty $I\subset \{1,\ldots,m\}$ we consider a sub-tuple $x_I$ that consists of all $x_i$ with $i\in I$. Consider some linear inequality for Kolmogorov complexities $\KS(x_I)$ for all $I$; we assume that it is split between two parts to make the coefficients positive:
\begin{equation}
\sum_{I\in \mathcal{I}} \lambda_I \KS(x_I) \le \sum_{J\in\mathcal{J}}\mu_J \KS(x_J) + O(\log n).
  \label{general}
\end{equation}
Here $\mathcal{I}$ and $\mathcal{J}$ are two disjoint families of subsets of $\{1,\ldots,m\}$, and $\lambda_I$ and $\mu_J$ are positive reals defined for $I\in\mathcal{I}$ and $J\in\mathcal{J}$.
Assume that this inequality is true for all~$n$ and for all tuples $\langle x_1,\ldots,x_m\rangle$ of $n$-bit strings (with a constant in $O(\log n)$-notation that does not depend on $n$ and $x_1,\ldots,x_m$). As he have mentioned, this assumption can be equivalently reformulated for entropies:
\begin{equation*}
\sum_{I\in \mathcal{I}} \lambda_I H(\xi_I) \le \sum_{J\in\mathcal{J}}\mu_J  H(\xi_J)
\end{equation*}
for every tuple $\langle\xi_1,\ldots,\xi_m\rangle$ of random variables (Romashchenko's theorem, see~\cite[Section 10.6, Theorem 211]{suv}).

Then we have the corresponding result about dimensions:

\begin{theorem}\label{direct}
Under these assumptions, for every set $S\subset \mathbb{R}^m$ and every non-negative reals $a_I$ \textup(defined for all $I\in\mathcal{I}$\textup) such that 
\[
\sum_{J\in\mathcal{J}}\mu_J \dim_p S_J \le \sum_{I\in\mathcal{I}}\lambda_I a_I,
\]
where $S_J$ is the projection of $S$ onto $J$-coordinates, there exist a splitting $S=\bigcup_{I\in \mathcal{I}}S^I$ such that
\[
\dim_H(S^I_I)\le a_I.
\]
\end{theorem}
The number of parts in the splitting is the same as the number of terms in the left-hand side of the inequality; they are indexed by $I\in\mathcal{I}$.  By $S^I_I$ we denote the $I$-projection of the part $S^I$; it is a set in $\mathbb{R}^k$ for $k=\#I$. The special case considered in Proposition~\ref{split1} has two terms on both sides of the inequality ($\#\mathcal{I}=\#\mathcal{J}=2$), and the coefficients $\lambda_I$ and $\mu_J$ are all equal to~$1$.

\begin{proof}
As before, consider some oracle $X$ that makes the effective packing dimensions of all $S_J$ for all $J\in\mathcal{J}$ minimal. Then we have
\[
\sum_{J\in\mathcal{J}}\mu_J \edim^X_p s_J\le \sum_{I\in\mathcal{I}}\lambda_I a_I,
\]
for every point $s=\langle s_1,\ldots,s_m\rangle\in S$; here $s_J$ stands for the projection of $s$ onto $J$-coordinates. The inequality for Kolmogorov complexities (that we assumed to be true) gives (after dividing by $n$ and taking the limit)
\[
\sum_{I\in \mathcal{I}}\lambda_I \edim^X_H s_I \le \sum_{J\in\mathcal{J}}\mu_J \edim^X_p s_J\left [\le \sum_{I\in\mathcal{I}}\lambda_I a_I\right],
\]
as before. (Note that the left hand side uses effective Hausdorff dimensions while the middle part uses effective packing dimensions, because of the limits.) This inequality is true for every point $s\in S$. Therefore, for every point $s\in S$ there exists some coordinate set $I\in\mathcal{I}$ such that
\[
\edim^X_H s_I \le a_I,
\]
and we can split the set $S$ according to these indices and get sets $S^I$ such that
\[
\edim^X_H s_I \le a_I
\]
for all points $s\in S^I$, and therefore 
\[
\dim_H S^I_I \le a_I,
\]
as required.
\end{proof}

\section{Equivalence}\label{equivalence}

We have shown that every (valid) linear inequality for entropies  can be translated to a statement about dimensions. In this section we show that this connection works in both directions:

\begin{theorem}\label{reverse}
If a linear inequality is not true for entropies, then the corresponding statement about dimensions, constructed as in Theorem~\textup{\ref{direct}}, is false.
\end{theorem}

\begin{proof}[Proof sketch]
We combine several well-known tools to achieve this result.

\textbf{1}. The first one is the characterization of inequalities in terms of the size of subgroups mentioned above. Let $G$ be some finite group, and let $H_1,\ldots,H_m$ be its subgroups. (We do not require them to be normal.) For every element $g\in G$ consider the cosets $g_1=gH_1$,\ldots, $g_m = gH_m$. If $g\in G$ is taken uniformly at random, the cosets $g_1,\ldots,g_m$ become (jointly distributed) random variables with common probability space $G$. Each $\xi_i$ is uniformly distributed on the family of all cosets $gH_i$; the size of this family is $\#G/\#H_i$, and the entropy of $\xi_i$ is $\log(\#G/\#H_i)$. 

We may consider tuples of them: let $g_I$ be the tuple of all $g_i$ with $i\in I$. It is easy to see that for every $I$ the values of $g_I$ correspond to cosets $gH_I$ for $H_I=\cap_{i\in I} H_i$, and the entropy $H(g_I)$ is $\log (\#G /\#H_I)$.

The result of Chan and Yeung~\cite{chan-yeung} says that the tuples of random variables constructed in this way are enough for testing inequalities: if an inequality 
\[
\sum_{I \in \mathcal{I}}  \lambda_I H(\xi_I) \le \sum_{J\in\mathcal{J}}\mu_J H(\xi_J).
\]
is not universally true (for all tuples of random variables $\langle \xi_1,\ldots,\xi_n\rangle$), then there exists a counterexample with groups, i.e., a finite group $G$ and its subgroups $H_1,\ldots,H_m$ that make the inequality false:
\[
\sum_{I \in \mathcal{I}}  \lambda_I H(g_I) > \sum_{J\in\mathcal{J}}\mu_J H(g_J).
\]
Note that the latter inequality can be reformulated in terms of sizes of a group, its subgroups and their intersections.

So we may assume that the inequality is not true for some group $G$ and its subgroups $H_i$, and use them to construct a counterexample that shows that the corresponding statement about dimensions is false.

\textbf{2}. For that we use standard results about the dimension of Cantor-type sets. Consider $N$-ary positional system where every real from $[0,1]$ is represented by an infinite sequence of digits $0\ldots N-1$. (As usual, the double representations for finite $N$-ary fractions do not matter much, and we ignore this problem.) Let $X$ be a subset of $\{0,\ldots,N-1\}$. Consider the set $C_X$ of all $N$-ary fractions with digits only in $X$ (for example, the classical Cantor set is $C_{\{0,2\}}$ for $N=3$). It is well known that $\dim_H(C_X)=\dim_p(C_X)=\log\#X/\log N$ (and this can be easily derived from the point-to-set principle). 

One can consider similarly defined sets in $\mathbb{R}^2, \mathbb{R}^3$ etc. For example, let $Y$ be a subset of $\{0,\ldots, N-1\}\times \{0,\ldots,N-1\}$. Then one can construct a set $C_Y\subset [0,1]\times [0,1]$ that consists of the pairs of $N$-ary fractions $(u_1u_2\ldots,v_1v_2\ldots)$ such that $(u_i,v_i)\in Y$ for every $Y$. The Hausdorff and packing dimensions of the set $C_Y$ are $\log\#Y /\log N$. 

For subsets of $[0,1]^m$ the construction goes as follows. Consider some set $A\subset \{0,\ldots,N-1\}^m$.  (Later, we let $m$ be the number of variables in the inequality we consider, and construct the set $A$ starting from the group $G$ and its subgroups $H_1,\ldots,H_m$.) Then construct the set $C_A \subset [0,1]^m$ such that $\langle x^1_1x^1_2\ldots,x^2_1x^2_2\ldots,\ldots,x^m_1x^m_2\ldots\rangle\in C_A$ if and only if $\langle x^1_i,\ldots,x^m_i\rangle\in A$ for all $i$. The dimension (packing or Hausdorff) of $C_A$ is $\log \#A/\log N$. 

The projection of the set $C_A$ on some set $I\subset\{1,\ldots,m\}$ of coordinates is the set $C_{A_I}$ of the same type that is constructed starting from the projection $A_I$ of $A$ onto the same coordinates. Therefore, to find the dimensions of all projections of $C_A$, we need to know only the size of the projections of $A$.

\textbf{3}. We can start this construction with a finite set $A\subset U_1\times\ldots\times U_m$ for arbitrary finite sets $U_1,\ldots,U_m$. Then we identify arbitrarily all $U_i$ with some subsets of $\{1,\ldots,N\}$ for large enough $N$, and construct the corresponding set $C_A\subset [0,1]^m$. For that we need that $\#U_i\le N$ for all $i$; the exact choice of $N$ is not important since the factor $1/\log N$ is the same for all the projections. 

Using this remark, we let $U_i$ be the range of $g_i$, i.e., the family of all cosets with respect to the subgroup $H_i$, and let 
\[
A = \{ (gH_1,\ldots,gH_m)\colon g\in G\}.
\]
Then, as we have seen, the dimension of $(C_A)_I$ is proportional to the logsize of the corresponding projection $A_I$, which equals the entropy of $g_I$:
\[
\dim (C_A)_I = \frac{H(g_I)}{\log N}.
\]
Therefore, we have
\[
\sum_{I \in \mathcal{I}}  \lambda_I \dim(C_A)_I > \sum_{J\in\mathcal{J}}\mu_J \dim(C_A)_J,
\]
assuming that we started with a counterexample to the inequality for entropies that involves group $G$ and subgroups $H_1,\ldots,H_m$. Note that we do not need to specify whether we consider packing or Hausdorff dimensions, since for our sets they are the same.

But this is not what we need: we need to show that a splitting of $C_A$ into sets with bounded dimensions of projections does not exist for some bounds $a_I$. Let us choose $a_I$ slightly smaller than $\dim(C_A)_I$ so that still
 \[
\sum_{I \in \mathcal{I}}  \lambda_I a_I > \sum_{J\in\mathcal{J}}\mu_J \dim(C_A)_J.
\]
We want to show that the assumption about dimensions is false for those $a_I$, namely, that one cannot split $C_A$ into a family of $C^I$ (for $I\in\mathcal{I}$) in such a way that
\[
\dim_H C^I_I \le a_I
\]
(here we have to specify the Hausdorff dimension since for the sets $C^I_I$ the Hausdorff and packing dimensions may differ). For that we note that the last inequality implies
\[
\dim_H C^I_I < \dim (C_A)_I
\]
due to the choice of $a_i$ that are smaller than $\dim (C_A)_I$. It remains to show that the last inequality implies
\[
\dim_H C^I < \dim C_A;
\]
then we get a contradiction, since the set $C_A$ cannot be represented as a finite union of sets of smaller Hausdorff dimensions.

To get the bound for $\dim_H C^I$ in terms of the dimension of its $I$-projection we use special properties of the set $A$ that corresponds to the group $G$ and its subgroups $H_1,\ldots,H_m$. Namely, for every set of indices $I$ the projection of $A$ onto $A_I$ is uniform (every element that has preimages has the same number of preimages; we already mentioned a similar property when saying that the variable $g_I$ is uniformly distributed on its image). Indeed, let $H=H_{\{1,\ldots,m\}}$ be the intersection of all subgroups: $H=\bigcap_{i=1,\ldots,m}H_i$, and let $H_I$ be the intersection of some of them: $H_I=\bigcap_{i\in I} H_i$. Then $H\subset H_I$ and we have surjective mappings:
\[
G \to G/H \to G/H_I.
\]
The projection of $A$ onto $A_I$ is the second mapping; the required property (the same number of preimages) is true since both the mapping $G\to G/H$ and the composition $G\to G/H_I$ have this property. The number of preimages for the projection is $\#H''/\#H'=\#A/\#A_I$.

We use this property and the following lemma.

\begin{lemma}
Assume that the projection $\pi_I \colon A \to A_I$ \textup(only $I$-coordinates remain\textup) is uniform. Consider the set $C_A$ and its projection $(C_A)_I$; let $d$ be the difference in their dimensions: $d=\dim C_A - \dim (C_A)_I$. Then, for every $X\subset C_A$ we have 
\[
\dim_H X \le  \dim_H X_I + d.
\]
\end{lemma}

Note that this lemma deals with two different projections: mappings $\pi_I\colon A \to A_I$ (finite sets) and $\Pi_I: C_A \to (C_A)_I$ (coordinate spaces); the second one applies the first one simultaneously for all positions in $N$-ary notation.

\begin{proof}[Proof of the lemma]
The dimension of $C_A$ is equal to $\log\#A/\log N$, and the dimension of $(C_A)_I=\Pi_I(C_A)$ is equal to $\log\#A_I/\log N$, so the difference is equal to 
\[
\log (\#A/\#A_I)/\log N.
\] 
The ratio $\#A/\#A_I$ is the size of preimages for the uniform projection $\pi_I\colon A\to A_I$. To specify the first $k$ digits in a point $x\in X$ we have to specify $k$ digits of its projection $x_I \in X_I$, and also for every of $k$ positions choose one of the preimages of some element of $A_I$. Now we may apply the point-to-set principle to get the desired result.
\end{proof}

The application of this lemma, as we have discussed, finishes the proof of Theorem~\ref{reverse}.
\end{proof}

\subsubsection*{Discussion}
Theorem~\ref{direct} applies the point-to-set principle to some type of statements in the dimension theory. Why these (rather exotic) statements could be interesting? There are two possible reasons.  

First, we get one more reason to consider the class of linear inequalities that are true for entropies of tuples: it can be equivalently characterized in terms of Kolmogorov complexity, in combinatorial terms (size of projections of multidimensional sets), as inequalities for group sizes --- and now in terms of dimensions. Second, the point-to-set principle was used to prove results about dimensions using algorithmic information theory. Theorem~\ref{reverse} shows that the reverse direction is also possible, at least in theory (it would be quite surprising to see a proof of some new inequality that goes this way).

\subsubsection*{Acknowledgments}

The author is grateful to all his colleagues, especially to Andrei Romashchenko, the members of the ESCAPE team in LIRMM and the Kolmogorov seminar, and the participants of the meetings organized by the American Institute of Mathematics and Dagstuhl in 2022 where some of the work presented here was discussed. Part of the work was supported by FLITTLA ANR-21-CE48-0023 grant.

\end{document}